\documentclass[reqno,letterpaper,11pt]{amsart}

\usepackage{hyperref}

\usepackage{palatino}

\usepackage{comment}

\usepackage{amsmath}
\usepackage{amsfonts,amssymb,amsmath,latexsym,wasysym,mathrsfs,bbm,stmaryrd}
\usepackage[all]{xy}
\usepackage[usenames]{color}
\usepackage{epsfig}

\usepackage{amsthm}
\usepackage{thmtools}

\newtheorem{theorem}{Theorem}[section]

\newtheorem{proposition}[theorem]{Proposition}
\newtheorem{definition}[theorem]{Definition}
\newtheorem{corollary}[theorem]{Corollary}

\newtheorem{lemma}[theorem]{Lemma}

\theoremstyle{remark}
\newtheorem{example}[theorem]{Example}
\newtheorem{remark}[theorem]{Remark}
\numberwithin{equation}{section}

\numberwithin{equation}{section}

\usepackage[mathcal]{euscript}
\usepackage{times}

\def\R{\mathbb R}
\def\C{\mathbb C}

\def\N{\mathbb N}

\def\e{\epsilon}

\def\1{\mathbbm{1}}
\def\tensor{\otimes}

\newcommand{\mx}[1]{\mathbf{#1}}

\long\def\symbolfootnote[#1]#2{\begingroup%
\def\thefootnote{\fnsymbol{footnote}}\footnote[#1]{#2}\endgroup}

\begin{document}

\title{Most Boson Quantum States are Almost Maximally Entangled}
\author{Shmuel Friedland}
\address{ Department of Mathematics, Statistics and Computer Science\\
University of Illinois at Chicago\\
Chicago, IL  60607-7045}
\email{friedlan@uic.edu}
\author{Todd Kemp}
\thanks{Kemp supported in part by NSF CAREER Award DMS-1254807}
\address{Department of Mathematics\\
University of California, San Diego \\
La Jolla, CA 92093-0112}
\email{tkemp@math.ucsd.edu}

\date{\today} 

\begin{abstract}  The geometric measure of entanglement $E$ of an $m$ qubit quantum state takes maximal possible value $m$.  In previous work of Gross, Flammia, and Eisert, it was shown that $E \ge m-O(\log m)$ with high probability as $m\to\infty$.  They showed, as a consequence, that the vast majority of states are too entangled to be computationally useful.  In this paper, we show that for $m$ qubit {\em Boson} quantum states (those that are actually available in current designs for quantum computers), the maximal possible geometric measure of entanglement is $\log_2 m$, opening the door to many computationally universal states.  We further show the corresponding concentration result that $E \ge \log_2 m - O(\log \log m)$ with high probability as $m\to\infty$.  We extend these results also to $m$-mode $n$-bit Boson quantum states.  \end{abstract}

\maketitle

\tableofcontents

\section{Introduction}

Quantum algorithms, in their ability to perform computations exponentially faster than what is strongly believed to be the maximum
speed of many classical algorithms, heavily depend on quantum entanglement.  For example, in \cite{Shor-Implement}, the authors  implemented a compiled version of Shor's quantum factoring algorithm in a photonic system, and observed high levels of entanglement.
It would be tempting to conclude that ``the more entanglement, the better'' when it comes to quantum computation.  In \cite{GFE09}, Gross, Flammia, and Eisert showed that this intuition is incorrect.  The {\em geometric measure of entanglement} $E$ is a monotone function on quantum states which takes values between $0$ (for product states) and $m$ in a system with $m$ qubits. (See Section \ref{section geometric measure} below for details).  Gross et.\ al.\ proved that if $\Psi$ is an $m$ qubit state with $E(\Psi)>m-\delta$, and if an NP problem can be solved by a computer with the power to perform local measurements on $\Psi$, then there is a purely classical algorithm that can solve the same problem in a time only approximately $2^\delta$ times longer.  Hence, any such states with $\delta = O(\log m)$ cannot be ``computationally universal''.  They then go on to show, remarkably, that the {\em vast majority} of quantum states have entanglement $E(\Psi)>m-O(\log m)$: letting $\mathbb{P}$ denote the Haar probability measure on the sphere of all $m$ qubit quantum states,
\begin{equation} \label{e.GFE} \mathbb{P}\Big(E(\Psi)\ge m-2\log_2(m)-3\Big)\ge 1-e^{-m^2}, \qquad \text{for }m\ge 11. \end{equation}
Hence, for large $m$, the proportion of states which may actually be used to gain more than a linear factor in performance is vanishingly small.

This situation may therefore seem dire, but the analysis ignores one glaring issue.  In any quantum computer based on photon interactions, all observed states are {\em Boson} quantum states (see Sections \ref{section symmetric tensors} and \ref{section Bosons}).  Bosons, or symmetric quantum states, form a small subspace of all states: the space of all $m$-mode tensors over $\C^2$ has dimension $2^m$, while the space of symmetric $m$-mode tensors over $\C^2$ has dimension $m+1$, exponentially smaller.  This has significant consequences.  For example, in a recent paper by the first coauthor and L. Wang, it is shown that the geometric entanglement of $m$ qubit Bosons is polynomially computable in $m$; cf.\ \cite{FW16}.  Presently, we show that the maximum possible geometric measure of entanglement of a Boson quantum state is much smaller than in the full space.

\begin{theorem} \label{t.max.entanglement} Let $n,m\ge 1$.  Denote $d_{n,m} = \binom{m+n-1}{m}$.  If $\Psi$ is an $m$ Boson quantum state on $\C^n$, then the geometric measure of entanglement of $\Psi$ satisfies
\[ 0\le E(\Psi) \le \log_2 d_{n,m}. \]
In particular, if $\Psi$ is an $m$ qubit Boson state, then $E(\Psi) \le \log_2(m+1)$. \end{theorem}
In particular, Gross, Flammia, and Eisert's argument about the usefulness of entangled states does not produce a pessimistic result here: since the the smallest $\delta>0$ for which $E(\Psi)>m-\delta$ is $\delta=O(m)$ for Boson states $\Psi$, the classical algorithm in \cite{GFE09} is exponentially slower than the quantum algorithm, as expected.

The main theorem of this paper addresses the proportion of Boson quantum states are that close to maximally entangled (even though this does not bear on their usefulness for computation).  We address the general question of Bosons over any finite-dimensional state space of dimension $\ge 2$.

\begin{theorem} \label{t.main} Let $n\ge 2$ and $m\ge 1$.  Denote $d_{n,m} = \binom{m+n-1}{m}$.  For fixed $n$,
\begin{equation} \label{e.main.1} \mathbb{P}\Big(E(\Psi) \ge \log_2 d_{n,m} - \log_2\log_2 d_{n,m} - 3\log_2 n - 1\Big) \ge 1-(d_{n,m})^{-n^3}  \end{equation}
for $m$ sufficiently large.   In particular, in the case $n=2$ where $\Psi$ is an $m$ qubit Boson state, we have
\begin{equation} \label{e.main.2} \mathbb{P}\Big(E(\Psi) \ge \log_2 m-\log_2\log_2 m-3\Big) \ge 1-\frac{1}{2m^{5/2}}, \qquad \text{for}\; m> 42. \end{equation}
\end{theorem} 

\begin{remark}\begin{itemize}
\item[(1)] The general condition on the size of $m$ to yield \eqref{e.main.1} is tedious to state.  Regardless, one finds that for all $m$, the probability is bounded below by $1-C(n)(d_{n,m})^{-n^3}$ for some constant $C(n)$ that does not depend on $m$. 
\item[(2)] The exponents $n^3$ and $5/2$ in \eqref{e.main.1} and \eqref{e.main.2} are not sharp.  In fact, the actual rate of decay is super-polynomial: the analysis in Section \ref{t.main} shows that, for any exponent $b>0$, there is a constant $a>0$ so that, for all $m$ sufficiently large,
\[  \mathbb{P}\Big(E(\Psi) \ge \log_2 d_{n,m} - \log_2\log_2 d_{n,m} - a\Big) \ge 1-(d_{n,m})^{-b}. \]
Note also that $d_{n,m} = O(m^{n-1})$, so the bounds could be stated in terms of super-polynomial decay in $m$ instead of $d_{n,m}$.
\item[(3)] However, we cannot prove super-exponential Gaussian-type concentration with this method, owing to the fact that the dimension $d_{n,m}$ of the symmetric tensor space is exponentially smaller than the dimension of the full tensor space.
\end{itemize}
\end{remark}

Our method of proof essentially follows \cite{GFE09}.  We use well-known concentration of measure for the Haar measure in high dimensional complex spheres, in conjunction with a sufficiently sharp bound on the cardinality of a net (of given tolerance) covering the sphere.  For this latter $\e$-net result, we give a very different proof from the one due to Gross et.\ al.  Much of the present work is to setup the problem in the restricted setting of symmetric tensors.  We now proceed to develop the proper background and notation required for this task.

\section{Background and Notation}

\subsection{Symmetric Tensors\label{section symmetric tensors}}

We work over a fixed finite dimensional complex vector space $V$, with the main object of interest being the {\em $m$-mode tensors} $V^{\tensor m}$.  There is a natural action of the symmetric group $\mathfrak{S}_m$ on $V^{\tensor m}$, given by the $\C$-linear extension of
\begin{equation} \label{e.Schur.rep} \sigma\cdot \tensor_{j=1}^m \mx{v}_j = \tensor_{j=1}^m \mx{v}_{\sigma^{-1}(j)}. \end{equation}
A tensor $T\in V^{\tensor m}$ is called {\bf symmetric} if $\sigma\cdot T = T$ for all $\sigma\in\mathfrak{S}_m$.  We denote
the subspace of symmetric tensors as $S^m(V)\subset V^{\tensor m}$.

There is a natural projection $P_m\colon V^{\tensor m}\to S^m(V)$; it is the $\C$-linear extension of
\[ P_m(\tensor_{j=1}^m \mx{v}_j) = \frac{1}{m!}\sum_{\sigma\in\mathfrak{S}_m} \sigma\cdot(\tensor_{j=1}^m \mx{v}_j). \]
Then $T\in V^{\tensor n}$ is symmetric iff $P_m(T) = T$.  We denote
\[ P_m(\tensor_{j=1}^m \mx{v}_j) = \odot_{j=1}^m \mx{v}_j \]
and so we may refer to $S^m(V)$ as $V^{\odot m}$.

Note that for any $\mx{v}\in V^{\tensor m}$, the rank-$1$ tensor $\mx{v}^{\tensor m}$ is symmetric: $\mx{v}^{\tensor m} = \mx{v}^{\odot m}$.  In fact, any symmetric tensor can be decomposed as a sum of rank-$1$ tensors.

\begin{proposition} \label{p.decomp.symmetric.rank1} For each $T\in S^m(V)$, there is a finite sequence $\{\mx{v}_1,\ldots,\mx{v}_r\}$ in $V$ such that
\begin{equation} \label{e.symmetric.tensor.rank} T = \sum_{j=1}^r (\mx{v}_j)^{\tensor m}. \end{equation}
\end{proposition}

Proposition \ref{p.decomp.symmetric.rank1} holds for any symmetric tensor with entries in an infinite field $\mathbb{F}$, as proved in \cite{AH95}.  Moreover, the above decomposition holds for tensors with entries in a finite field $\mathbb{F}$ provided $\#\mathbb{F}\ge m$; cf.\ \cite[Proposition 7.2]{FS13}.

\begin{remark} The minimal $r$ which can be used in \eqref{e.symmetric.tensor.rank} is called the {\em symmetric tensor rank}
of $T$.  The {\em Comon Conjecture} states that this is equal to the usual rank of $T$ in general;
it is only known to hold in certain special cases; cf.\ \cite{Frsymrank15} .
\end{remark}

Now, fix a basis $\{\mx{e}_j\}_{j=1}^n$ of $V$; then we can expand any tensor $T\in V^{\tensor n}$ in terms of the tensor basis
\[ \{\mx{e}_{\tensor\mx{j}}\colon \mx{j}\in[n]^m\} \]
where if $\mx{j}=(j_1,\ldots,j_m)\in[n]^m$ then $\mx{e}_{\tensor\mx{j}} = \tensor_{k=1}^m \mx{e}_{j_k}$.  The symmetric projections of these basis tensors, $\mx{e}_{\odot\mx{j}} = P_m(\mx{e}_{\tensor\mx{j}}) = \odot_{k=1}^m \mx{e}_{j_k}$ are not
linearly independent, since $\sigma\cdot\mx{e}_{\odot\mx{j}} = \mx{e}_{\odot\mx{j}}$ for all $\sigma\in\mathfrak{S}_m$.  To
generate a basis, we consider only nondecreasing indices:
\[ [n]^{\uparrow m} = \{\mx{i}=(i_1,\ldots,i_m)\in[n]^m\colon i_1\le \cdots \le i_m\}. \]
The set $\{\mx{e}_{\odot\mx{i}}\colon\mx{i}\in[n]^{\uparrow m}\}$ is a basis for $S^m(V)$, and hence
\begin{equation} \label{e.dimension} \mathrm{dim}_{\C}(S^m(V)) = \#([n]^{\uparrow m}) = \textstyle{\binom{n+m-1}{m}}:=d_{n,m}. \end{equation}

Given any $\mx{j}\in[n]^m$, there is a unique $\mx{i}\in[n]^{\uparrow m}$ such that, for some $\sigma\in\mathfrak{S}_m$, $\sigma\cdot\mx{j} = \mx{i}$.  (There may be several $\sigma$ that work here, but there is only one $\mx{i}\in[n]^{\uparrow m}$.)  Denote
this unique element as $\mx{i} = \uparrow\!(\mx{j})$.  Let us define
\[ c(\mx{i}) = \#\{\mx{j}\in[n]^m\colon \uparrow\!(\mx{j}) = \mx{i}\}. \]
This coefficient can be computed thus: the index $\mx{i}$ induces a set partition $\pi(\mx{i})$ of $[m]$, where $k\sim_{\pi(\mx{i})}\ell$ iff $i_k=i_\ell$.  (Thus, if $\uparrow\!(\mx{j})=\mx{i}$, then $\pi(\mx{j})=\pi(\mx{i})$.)  If $\pi(\mx{i})$ has blocks of sizes $m_1,m_2,\ldots,m_b$, then
\begin{equation} \label{e.c(i)} c(\mx{i}) = \frac{m!}{m_1!\cdots m_b!}. \end{equation}

Let $T\in S^m(V)$, and expand $T$ in terms of the standard basis $\{\mx{e}_{\tensor\mx{i}}\colon\mx{i}\in[n]^m\}$ of the
full tensor space $V^{\tensor m}$, and also in terms of the basis $\{\mx{e}_{\odot\mx{i}}\colon\mx{i}\in[n]^{\uparrow m}\}$ of
$S^m(V)$:
\[ T = \sum_{\mx{j}\in[n]^m} T_{\mx{j}}\,\mx{e}_{\tensor\mx{j}}, \qquad T = \sum_{\mx{i}\in[n]^{\uparrow m}} T_{\mx{i}}'\,\mx{e}_{\odot\mx{i}}. \]
We can express the coefficients $T_{\mx{i}}'$ in terms of the coefficients $T_{\mx{j}}$, as follows.  Since $T\in S^m(V)$,
$T_{\mx{j}} = T_{\uparrow\!(\mx{j})}$ for all $\mx{j}\in[n]^m$.  Thus
\begin{align*} T = \sum_{\mx{j}\in[n]^m} T_{\mx{j}}\,\mx{e}_{\tensor\mx{j}} = \sum_{\mx{i}\in[n]^{\uparrow m}}\sum_{\mx{j}\in[n]^m\atop \uparrow(\mx{j})=\mx{i}} T_{\mx{j}}\,\mx{e}_{\tensor\mx{j}}
=  \sum_{\mx{i}\in[n]^{\uparrow m}} T_{\mx{i}} \sum_{\mx{j}\in[n]^m\atop \uparrow(\mx{j})=\mx{i}} \mx{e}_{\tensor\mx{j}}
 = \sum_{\mx{i}\in[n]^{\uparrow m}} T_{\mx{i}}\, c(\mx{i})\,\mx{e}_{\odot\mx{i}}.
\end{align*}
We conclude that
\[ T_{\mx{i}}' = c(\mx{i})T_{\mx{i}}. \]

\subsection{The Hilbert--Schmidt inner product}

Now we fix an inner product $\langle\cdot,\cdot\rangle_V$ on $V$.  This induces an inner product on $V^{\tensor m}$, typically called the {\em Hilbert--Schmidt} inner product, denoted $\langle\cdot,\cdot\rangle_2$.  It is the unique sesquilinear extension of
\[ \langle \tensor_{j=1}^m \mx{v}_j,\tensor_{j=1}^m \mx{w}_j\rangle_2 = \prod_{j=1}^m \langle \mx{v}_j,\mx{w}_j\rangle_V. \]
The basis $\{\mx{e}_{\tensor\mx{j}}\colon \mx{j}\in[n]^m\}$ is orthonormal with respect to the Hilbert--Schmidt inner product; hence for $S,T\in V^{\tensor m}$,
\[ \langle S,T\rangle_2 = \sum_{\mx{j}\in[n]^m} \overline{S_{\mx{j}}}T_{\mx{j}}. \]
If $S,T\in S^m(V)$, the Hilbert--Schmidt inner product can be written in terms of the symmetric coefficients as
\[ \langle S,T\rangle_2 = \Big\langle \sum_{\mx{i}\in[n]^{\uparrow m}} S'_{\mx{i}}\,\mx{e}_{\odot\mx{i}}, \sum_{\mx{j}\in[n]^{\uparrow m}} T'_{\mx{j}}\,\mx{e}_{\odot\mx{j}}\Big\rangle_2 = \sum_{\mx{i},\mx{j}\in[n]^{\uparrow m}} \overline{S'_{\mx{i}}}T'_{\mx{j}}\, \langle \mx{e}_{\odot\mx{i}},\mx{e}_{\odot\mx{j}}\rangle_2. \]
If $\mx{i},\mx{j}\in[n]^{\uparrow m}$ are not equal, then for any $\sigma\in\mathfrak{S}_m$ there will be some index $k\in[m]$ with
$i_k\ne j_{\sigma(k)}$.  Since the original basis is orthonormal, it follows that $\langle \mx{e}_{\odot\mx{i}},\mx{e}_{\odot\mx{j}}\rangle_2=0$.  On the other hand, if $\mx{i}=\mx{j}$, then it is straightforward to compute that
$\langle \mx{e}_{\odot\mx{i}},\mx{e}_{\odot\mx{i}}\rangle_2 = \frac{1}{c(\mx{i})}$.  Hence, we find that
\begin{equation} \label{e.HS.inner.product.symmetric} \langle S,T\rangle_2 = \sum_{\mx{i}\in[n]^{\uparrow m}} \frac{1}{c(\mx{i})}\, \overline{S'_{\mx{i}}}T'_{\mx{i}}. \end{equation}
In particular, this shows that the basis vectors $\{\mx{e}_{\odot\mx{i}}\colon\mx{i}\in[n]^{\uparrow m}\}$ are orthogonal,
but not generally normalized: $\|\mx{e}_{\odot\mx{i}}\|_2= c(\mx{i})^{-1/2}$.  We can then normalize them
\begin{equation} \label{e.hat.e.i} \hat{\mx{e}}_{\odot\mx{i}} = \sqrt{c(\mx{i})}\,\mx{e}_{\odot\mx{i}} \end{equation}
to produce an orthonormal basis of $S^m(V)$.  For $m=2$ the above basis is called Dicke states; c.f.\ \cite{Dic}.

\subsection{The Spectral Norm}

In the special case $m=2$, $V^{\tensor 2}$ can be identified with $\mathrm{End}(V)$ (using an inner product on $V$) by
the linear extension of the map $\mx{u}\tensor\mx{v}\mapsto \mx{u}\mx{v}^\ast$.  (In physics notation, this rank-$1$ linear
transformation is usually denoted  $|\mx{u}\rangle\langle\mx{v}|$.)  Under
this identification, the Hilbert--Schmidt inner product introduced above gives the usual Hilbert--Schmidt inner product on matrices: $\langle A,B\rangle_2 =  \sum_{i,j=1}^n \bar{A}_{ij}B_{ij} = \mathrm{Tr}(A^\ast B)$.  While easy to compute, the resulting norm $\|\cdot\|_2$ on matrices is not as widely useful as the {\em spectral norm} of $A$:
\[ \|A\|_{\infty} = \max\{\|A\mx{v}\|\colon \|\mx{v}\| = 1\}. \]
(For short, we denote the inner product $\langle\cdot,\cdot\rangle_V$ simple as $\langle\cdot,\cdot\rangle$; the corresponding norm on $V$ is thus denoted $\|\cdot\|$.)

\begin{remark} Another, perhaps more common, term used is {\em operator norm}.  It is called the spectral norm because it is the modulus of the largest singular value of $A$, or the largest eigenvalue of $\sqrt{A^\ast A}$. \end{remark}

A nice alternative way to compute the spectral norm is as
\[ \|A\|_{\infty} = \max\{|\langle \mx{u},A\mx{v}\rangle|\colon \|\mx{u}\|=\|\mx{v}\|=1\}. \]
(The Cauchy--Schwarz inequality shows that $|\langle \mx{u},A\mx{v}\rangle|\le \|A\|_\infty$; by taking $\mx{u} = A\mx{v}/\|A\mx{v}\|$,
we see that the maximum is achieved.)  At the same time, note that
\[ \langle \mx{u}\mx{v}^\ast, A\rangle_2 = \mathrm{Tr}((\mx{u}\mx{v}^\ast)^\ast A) = \mathrm{Tr}(\mx{v}\mx{u}^\ast A) = \mathrm{Tr}(A\mx{v}\mx{u}^\ast) = \langle \mx{u},A\mx{v}\rangle. \]
Hence, using the above identification, a tensor $T\in V^{\tensor 2}$, identified as a matrix, has spectral norm
\[ \|T\|_\infty = \max\{|\langle T,\mx{u}\tensor\mx{v}\rangle_2|\colon \|\mx{u}\|=\|\mx{v}\|=1\}. \]
This prompts the following general definition.

\begin{definition} The {\bf spectral norm} $\|\cdot\|_\infty$ on $V^{\tensor m}$ is defined by
\begin{equation} \label{e.spectral.norm} \|T\|_\infty = \max\{|\langle T,\tensor_{j=1}^m \mx{v}_j\rangle_2|\colon \|\mx{v}_1\|=\cdots=\|\mx{v}_m\|=1\}. \end{equation}
\end{definition}
If $m=1$, then $\|\mx{v}\|_\infty = \|\mx{v}\|$ for $\mx{v}\in V^{\tensor 1} = V$.  As shown above, if $m=2$, the spectral norm
corresponds to the spectral norm of matrices under the usual identification of $2$-mode tensors as matrices.  In general, the spectral
norm is a {\em tensor norm}: it satisfies
\begin{equation} \label{e.tensor.norm} \|\tensor_{j=1}^m \mx{v}_j\|_\infty = \prod_{j=1}^m \|\mx{v}_j\| \end{equation}
as can be quickly verified from the definition.  As a norm on a finite dimensional vector space, it is equivalent to all other
norms, including the Hilbert--Schmidt norm.  The distortion in this comparison is as follows.

\begin{lemma} \label{l.norm.inequality} Let $V$ be an $n$-dimensional inner product space.  Then the Hilbert--Schmidt norm $\|\cdot\|_2$ and spectral norm $\|\cdot\|_\infty$ on $V^{\tensor m}$ satisfy
\[ n^{-m/2}\|T\|_2 \le \|T\|_\infty \le \|T\|_2, \qquad \forall\; T\in V^{\tensor m}. \]
\end{lemma}

\begin{proof} The Cauchy-Schwarz inequality implies that 
\[|\langle T,\tensor_{j=1}^m \mx{v}_j\rangle_2|\le \|T\|_2\|\tensor_{j=1}^m \mx{v}_j\|_2= \|T\|_2\prod_{j=1}^m \|\mx{v}_j\|.\] 
This yields the inequality  $\|T\|_\infty \le \|T\|_2$.  Equality holds if and only if $T$ is a rank one tensor.  Observe next
\[\|T\|_2^2=\sum_{\mx{j}\in[n]^m} |T_{\mx{j}}|^2 \le n^m \max\{|T_{\mx{j}}|^2:\mx{j}\in[n]^m\}. \]
Note that $|T_{\mx{j}}|^2 = |\langle T,\mx{e}_{\tensor\mx{j}}\rangle_2|^2 \le \|T\|_{\infty}^2$.  This yields the inequality $n^{-m/2}\|T\|_2 \le \|T\|_\infty$. 
\end{proof}

\begin{remark} The same argument cannot be applied in the symmetric case. The tensors $\mx{e}_{\tensor\mx{j}}$ are product states and so may be used in the computation of $\|T\|_{\infty}$, while the tensors $\hat{\mx{e}}_{\odot\mx{j}}$ are {\em not} product states in general.   The proof of Theorem \ref{t.max.entanglement} requires a much more involved argument. \end{remark}

If the tensor $T\in V^{\tensor m}$ happens to be symmetric, $T\in S^m(V)$, then the rank-$1$ tensor in \eqref{e.spectral.norm} can also be
taken in $S^m(V)$, which means it must be of the form $\mx{v}^{\odot m} = \mx{v}^{\tensor m}$ for some unit vector $\mx{v}$.  This is Banach's theorem.

\begin{theorem}[Banach's Theorem, \cite{Ban38}] \label{t.Banach} If $T\in S^m(V)$, then
\[ \|T\|_\infty = \max\{|\langle T,\mx{v}^{\tensor m}\rangle_2|\colon \|\mx{v}\|=1\}. \]
\end{theorem}

\begin{example} \label{eg.1} We can explicitly compute the spectral norm of any unit symmetric basis tensor $\hat{\mx{e}}_{\odot\mx{i}}$ (as defined in \eqref{e.hat.e.i}).  First note that, if $\mx{i}=(i,\ldots,i)$ then $\hat{\mx{e}}_{\mx{i}}=\mx{e}_i^{\tensor m}$, and so $\|\hat{\mx{e}}_{\mx{i}}\|_\infty = \|\mx{e}_i\|^m = 1$ by \eqref{e.tensor.norm}.  By Lemma \ref{l.norm.inequality},
these basis elements have maximal spectral norm in the unit ball in $V^{\tensor m}$.

More generally, let $\mx{i}=(i_1,\ldots,i_m)$.  Before normalizing we have $\mx{e}_{\odot\mx{i}} = \frac{1}{m!}\sum_{\sigma\in\mathfrak{S}_m} \sigma\cdot \mx{e}_{\tensor\mx{i}}$.  Then
\[ \langle \mx{e}_{\odot\mx{i}},\mx{v}^{\tensor m}\rangle_2 = \frac{1}{m!}\sum_{\sigma\in\mathfrak{S}_m} \langle \sigma\cdot\mx{e}_{\tensor\mx{i}},\mx{v}^{\tensor m}\rangle_2 = \frac{1}{m!}\sum_{\sigma\in\mathfrak{S}_m} \prod_{k=1}^m \langle \mx{e}_{i_{\sigma^{-1}(k)}},\mx{v}\rangle = \frac{1}{m!}\sum_{\sigma\in\mathfrak{S}_m} \prod_{k=1}^m v_{i_{\sigma^{-1}(k)}} \]
where $v_i = \langle \mx{e}_i,\mx{v}\rangle$.  Since multiplication of complex numbers is commutative, all of the terms in this sum are equal, and so we simply have
\[ \langle \mx{e}_{\odot\mx{i}},\mx{v}^{\tensor m}\rangle_2 = v_{i_1}\cdots v_{i_m} \]
and so
\[ \|\mx{e}_{\odot\mx{i}}\|_{\infty} = \max\{|v_{i_1}\cdots v_{i_m}| \colon |v_1|^2+\cdots+|v_m|^2=1\}. \]
Computing this maximum is a matter of elementary calculus.  The result is as follows: if $\pi(\mx{i})$ is a partition
with blocks of sizes $m_1,\ldots,m_b>0$ (where $m_1+\cdots+m_b=m$), then
\begin{equation} \label{e.spec.norm.ei} \|\mx{e}_{\odot\mx{i}}\|_{\infty} = \sqrt{\frac{m_1^{m_1}\cdots m_b^{m_b}}{b^m}}. \end{equation}
The normalization coefficient $c(\mx{i})$ in this case is $c(\mx{i}) = \frac{m!}{m_1!\cdots m_b!}$ (cf. \eqref{e.c(i)}), and so
\begin{equation} \label{e.spec.norm.ei2} \|\hat{\mx{e}}_{\odot\mx{i}}\|_{\infty} = \sqrt{c(\mx{i})}\|\mx{e}_{\odot\mx{i}}\|_{\infty} = \sqrt{\frac{m!}{m^m}}\cdot\prod_{j=1}^b \sqrt{\frac{m_j^{m_j}}{m_j!}}. \end{equation}
\end{example}

\subsection{Geometric Measure of Entanglement\label{section geometric measure}}

Let $\mathscr{P}_1(V^{\tensor n})$ denote the set of (unit-length) product states:
\[ \mathscr{P}_1(V^{\tensor n}) = \{\mx{v}_1\tensor\cdots\tensor \mx{v}_n\colon \|\mx{v}_1\|=\cdots=\|\mx{v}_n\|=1\}. \]
The {\em geometric measure of entanglement} of a (unit length) tensor $T\in V^{\tensor n}$ can be defined to be the distance from $T$ to $\mathscr{P}_1(V^{\tensor n})$,
\[ \inf_{S\in\mathscr{P}_1(V^{\tensor n})} \|S-T\|_2. \]
This quantity is $0$ iff $T$ is a product state, thus capturing how entangled $T$ is.  Squaring it and expanding, and using phase invariance, one sees that it can be expressed easily in terms of the spectral norm:
\[ \inf_{S\in\mathscr{P}_1(V^{\tensor n})} \|S-T\|_2^2 = 2-2\|T\|_{\infty}. \]
In particular, we see that $\|T\|_{\infty}=1$ if and only if $T$ is a product state.  Gross et.\ al.\ \cite{GFE09} dispensed with the distance measure and instead redefined the geometric measure of entanglement as
\[ E(T) = -2\log_2\|T\|_{\infty}. \]
The $\log_2$ makes sense from an entropy point of view, and restores the property that $E(T)=0$ if and only if $T$ is a product state.  Lemma \ref{l.norm.inequality} then shows that, for a unit-length tensor $T$,
\[ 0\le E(T) \le (\log_2 n)m. \]
In particular, for $m$ qubit states, the possible range of the geometric measure of entanglement is $[0,m]$.

\begin{example} \label{eg.2} Continuing Example \ref{eg.1}, let $\mx{i}$ be a multi-index whose partition $\pi(\mx{i})$ has blocks of sizes $m_1,\ldots,m_b$, where $m_1+\cdots+m_b=m$.  Using Stirling's approximation with \eqref{e.spec.norm.ei} yields
\begin{equation} \label{e.eg.2} E(\hat{\mx{e}}_{\odot\mx{i}}) = \frac12\left(\sum_{j=1}^b \log_2 m_j - \log_2 m\right) + O(b) \end{equation}
where the $O(b)$ constant is in the interval $[b\log_2\sqrt{2\pi}-\log_2 e,b\log_2e - \log_2\sqrt{2\pi}]\approx[1.32b-1.45,1.45b-1.32]$.  Consider the two extreme cases: when $b=1$ (so all indices in $\mx{i}$ are equal) and when $b=m$ (so all indices in $\mx{i}$ are distinct).  In the former case, $\hat{\mx{e}}_{\odot\mx{i}}$ is a product state and $E(\hat{\mx{e}}_{\odot\mx{i}})=0$; in the latter case, each $m_j=1$ and so $E(\hat{\mx{e}}_{\odot\mx{i}}) = (\log_2 e)m-\frac12\log_2 m +O(1)$ (compare to the maximum possible value $(\log_2 n)m$, where in this case $n\ge m$).  Generally speaking, the fewer coincidences among the indices of $\mx{i}$, the greater the entanglement of $\hat{\mx{e}}_{\odot\mx{i}}$.

In the case of $m$ qubit Bosons, the basis states are $\hat{\mx{e}}_{(j,m-j)}$; except for the pure states with $j=0,m$, we have $b=2$.
Then \eqref{e.eg.2} yields $E(\hat{\mx{e}}_{(j,m-j)}) = \frac12\log_2(j(m-j))-\log_2m +O(1)$.  This is maximized at $j=\frac{m}{2}$, yielding the precise estimate $E(\hat{\mx{e}}_{(j,m-j)})\le \frac12\log_2 m +c$ where $0.20<c<0.56$; this is a factor of $2$ smaller than the typical value; cf.\ Theorem \ref{t.main}.
\end{example}

\subsection{Boson Quantum States\label{section Bosons}}

Quantum states in an $m$-partite system are not exactly given by unit-length tensors in $V^{\tensor n}$.  Two tensors that are equal up to a complex phase factor represent the same quantum state.  This is true in the symmetric case as well.  The set of {\bf Boson quantum states} $\mathcal{B}^m(V)$ is the quotient of the unit sphere $S^m_1(V)$ in $S^m(V)$ by the relation $T\sim \zeta T$ for all $\zeta$ in the unit circle in $\C$:
\[ \mathcal{B}^m(V) = \{T\in S^m_1(V)\colon T\sim \zeta T \; \forall \, |\zeta|=1\}. \]
We generally denote Boson states in $\mathcal{B}^m(V)$ using upper-case Greek letters $\Psi$ and $\Phi$, and reserve $T$ for tensors (not modding out by phase factors).
 
Note that the spectral norm \eqref{e.spectral.norm} is invariant under multiplication by a complex phase, and so it descends to $\mathcal{B}^m(V)$; similarly, the geometric measure of entanglement also descends to $\mathcal{B}^m(V)$. Nonetheless, we must be a little careful treating $\mathcal{B}^m(V)$ as a metric space, since the distance between two {\em distinct} quantum states is not well-defined in terms of the distance between two representative tensors (that can each be multiplied by independent phase factors).  We therefore define, for two states,
\begin{equation} \label{e.def.metric} \|\Psi-\Phi\|_2:= \min\{\|S-T\|_2\colon S\in\Psi, T\in\Phi\}. \end{equation}
That is, fixing any two representative tensors $S_0\in\Psi$ and $T_0\in\Phi$,
\[ \|\Psi-\Phi\|_2=\min_{|\zeta|=|\eta|=1}\|\eta S_0 - \zeta T_0\|_2 = \min_{|\zeta|=1} \|S_0-\zeta T_0\|_2 \]
It is straightforward to check that this makes $\mathcal{B}^m(V)$ into a compact metric space; the distance function evidently satisfies the triangle inequality, and it yields $0$ if and only if $S_0=\zeta T_0$ for some $|\zeta|=1$, which is precisely to say that $\Phi=\Psi$.

We now introduce the Haar probability measure on $\mathcal{B}^m(V)$ we use in Theorem \ref{t.main}.  The Hilbert--Schmidt inner product on $S^m(V)$ identifies $S_1^m(V)$ with a unit sphere.  To be definite, we use the orthonormal basis $\{\hat{\mx{e}}_{\odot\mx{i}}\colon\mx{i}\in[n]^{\uparrow m}\}$ to identify $S_1^m(V)$ isometrically with the unit sphere in $\C^{[n]^{\uparrow m}}$:
\[ \C^{[n]^{\uparrow m}}\ni (z_{\mx{i}})\mapsto \sum_{\mx{i}\in[n]^{\uparrow m}} z_{\mx{i}} \hat{\mx{e}}_{\odot\mx{i}}. \]
Letting $d_{n,m} = \#[n]^{\uparrow m} = \binom{n+m-1}{m}$ \eqref{e.dimension}, this means $S_1^m(V)$ is isometrically isomorphic to the sphere $\mathbb{S}^{2d_{n,m}-1}$.  Since the uniform measure on the sphere is invariant under rotations, any other identification (i.e.\ choice of orthonormal basis) would yield the same measure.

Thus, there is a surjective linear map
\[ \mathbb{S}^{2d_{n,m}-1}\to S_1^m(V) \to \mathcal{B}^m(V) \]
given by composing the above isomorphism with the natural projection map $S^m_1(V)\to \mathcal{B}^m(V)$.
We refer to the push forward of the uniform probability measure on $\mathbb{S}^{2d_{n,m}-1}$ to $\mathcal{B}^m(V)$ as the {\em Haar measure on Bosons}.  Note: the uniform measure on $\mathbb{S}^{2d_{n,m}-1}$ is invariant under the map $T\mapsto\zeta T$ for any phase $\zeta$, and so the Haar measure on $\mathcal{B}^m(V)$ is essentially the same as the Haar measure on $S_1^m(V)$.


\section{Proofs of the Theorems}

\subsection{Proof of Theorem \ref{t.max.entanglement}}


Since $E(\Psi)\ge 0$ for all $\Psi$, our goal is to prove that $-2\log_2\|\Psi\|_{\infty} \le \log_2 \binom{n+m-1}{m}$, or equivalently
\begin{equation} \label{e.E.lower.bound} \|\Psi\|_{\infty}^2 \ge \frac{1}{\binom{n+m-1}{m}}, \qquad \forall \,\Psi\in \mathcal{B}^m(\C^n). \end{equation}
Since the spectral norm is invariant under multiplication by a phase, we may work directly with symmetric tensors $T$ (instead of Boson quantum states $\Psi$ in the quotient space).

The key result needed is the following invariance statement, which can be found as \cite[Lemma 4.3.1]{Ren05}.

\begin{proposition} \label{p.Schur} Let $\mathbb{S}^1(\C^n)$ denote the sphere $\{\mx{v}\in\C^n\colon \|\mx{v}\|=1\}$ in $\C^n$.  Let $P_{\mx{v}}$ denote the orthogonal projection operator from $\C^n$ onto the span of $\mx{v}$; in physics notation, $P_{\mx{v}} = |\mx{v}\rangle\langle \mx{v}|$.  Then for any symmetric tensor $T\in S^m(\C^n)$,
\begin{equation} \label{e.Schur} \int_{\mathbb{S}^1(\C^n)} P_{\mx{v}}^{\tensor m}(T)\,d\mx{v} = \textstyle{\binom{n+m-1}{m}}^{-1}T \end{equation}
where the integral is taken with respect to the Haar probability measure on $\mathbb{S}^1(\C^n)$.
\end{proposition}
In \cite{Ren05}, this is proved by direct calculation.  We give an independent proof here that is based on representation theory.  We first need the following (well-known) lemma.

\begin{lemma} \label{l.irreducible} Let $\rho_m\colon\mathrm{U}(n)\to \mathrm{GL}(S^m(\C^n))$ denote the complex representation given by
\[ \rho_m(U)(\mx{v}_1\odot\cdots\odot \mx{v}_m) = (U\mx{v}_1)\odot\cdots\odot(U\mx{v}_m). \]
Then $\rho_m$ is irreducible. \end{lemma}

\begin{proof} First, note that $\rho_m$ extends (by the same formula) to a complex representation of $\mathrm{GL}(n,\C)$.  It is one of the statements of the Schur--Weyl duality that this representation---the $m$-fold symmetric tensor power of the standard representation of $\mathrm{GL}(n,\C)$---is irreducible; see, for example, \cite[Theorem 6.3(4)]{Fulton-Harris}.

To conclude the proof, we note that any irreducible complex representation of $\mathrm{GL}(n,\C)$ restricts to an irreducible representation of $\mathrm{U}(n)$.  Indeed, since both groups are connected, irreduciblity of a group representation $\rho$ is equivalent to irreducibility of the associated Lie algebra representation $d\rho$.  But the Lie algebra $\mathfrak{gl}(n,\C)$ of $\mathrm{GL}(n,\C)$ is the complexification of the Lie algebra $\mathfrak{u}(n)$ of $\mathrm{U}(n)$: $\mathfrak{gl}(n,\C) = \mathfrak{u}(n)\oplus i\mathfrak{u}(n)$.  Hence, any complex subspace invariant under $d\rho|_{\mathfrak{u}(n)}$ is automatically invariant under $d\rho$.  \end{proof}

\begin{proof}[Proof of Proposition \ref{p.Schur}] For $T\in S^m(\C^n)$, let $\mathcal{R}(T)$ denote the left-hand-side of \eqref{e.Schur}; so $\mathcal{R}\in\mathrm{End}(S^m(\C^n))$.  We will show that $\mathcal{R}$ is a constant multiple of the identity on $S^m(\C^n)$.

Let $\rho_m\colon\mathrm{U}(n)\to \mathrm{GL}(S^m(\C^n))$ be the irreducible representation in Lemma \ref{l.irreducible}. It is a simple matter to compute that, for $\mx{v},\mx{w}\in\C^n$ and $U\in\mathrm{U}(n)$,
\[ P_{\mx{v}}^{\tensor m}\rho_m(U)(\mx{w}^{\tensor m}) = \langle \mx{v},U\mx{w}\rangle^m \mx{v}^{\tensor m} = \langle U^{\ast}\mx{v},\mx{w}\rangle^m \mx{v}^{\tensor m}.  \]
Now using the unitary invariance of the Haar probability measure, we have
\[ \mathcal{R}\rho_m(U)(\mx{w}^{\tensor m}) = \int_{\mathbb{S}^1(\C^n)} \langle U^\ast\mx{v},\mx{w}\rangle^m\mx{v}^{\tensor m}\,d\mx{v} = \int_{\mathbb{S}^1(\C^n)} \langle \mx{v},\mx{w}\rangle^m(U\mx{v})^{\tensor m}\,d\mx{v}.  \]
On the other hand, note that
\[ \rho_m(U)P_{\mx{v}}^{\tensor m}(\mx{w}^{\tensor m}) = \langle \mx{v},\mx{w}\rangle^m \rho_m(U)(\mx{v}^{\tensor m})  = \langle \mx{v},\mx{w}\rangle^m (U\mx{v})^{\tensor m}. \]
Passing the integral through the linear operator $\rho_m(U)$, we conclude that $\mathcal{R}\rho_m(U)(\mx{w}^{\tensor m}) = \rho_m(U)\mathcal{R}(\mx{w}^{\tensor m})$ for each $U\in\mathbb{U}(n)$.  By Proposition \ref{p.decomp.symmetric.rank1}, any $T\in S^m(\C^n)$ is a linear combination of rank $1$ tensors of the form $\mx{w}^{\tensor m}$ for some $\mx{w}\in\C^n$; thus, we conclude that the irreducible representation $\rho_m$ commutes with $\mathcal{R}$ on $S^m(\C^n)$.  By Schur's lemma, it follows that $\mathcal{R}=c\cdot \mathrm{Id}_{S^m(\C^n)}$ for some constant $c\in\C$.

To compute the constant $c$, we follow \cite{Ren05} and use the fact that $P_{\mx{v}}$ is a rank-$1$ projection, hence has trace $1$.  Thus $\mathrm{Tr}(P_{\mx{v}}^{\tensor m}) = \mathrm{Tr}(P_{\mx{v}})^m=1$, and so
\[ \mathrm{Tr}(\mathcal{R}) = \int_{\mathbb{S}^1(\C^n)} \mathrm{Tr}(P_{\mx{v}})^m\,d\mx{v} = \int_{\mathbb{S}^1(\C^n)} d\mx{v} = 1. \]
Thus $1 = \mathrm{Tr}(\mathcal{R}) = \mathrm{Tr}(c\cdot \mathrm{Id}_{S^m(\C^n)}) = c\cdot \mathrm{dim}_{\C}(S^m(\C^n))$.  The result now follows from \eqref{e.dimension}. \end{proof}

\begin{remark} It is important to note that the preceding proof fundamentally requires the underlying vector space to be complex: both for the use of the Schur--Weyl duality in Lemma \ref{l.irreducible}, and the use of Schur's lemma in Proposition \ref{p.Schur}.  In fact, the result is {\em not true} in the real setting.  This can be seen by giving a more direct computational proof of the proposition, computing the integral in spherical coordinates.  In that case, in the $\R^n$ setting, the integral $\mathcal{R}(T)$ is equal to $cT$ plus $\lfloor n/2 \rfloor$ additional lower-order terms: contractions of $T$ with strictly positive coefficients.  \end{remark}

\begin{proof}[Proof of Theorem \ref{t.max.entanglement}] We proceed to prove \eqref{e.E.lower.bound}; as discussed above, this suffices to prove Theorem \ref{t.max.entanglement}.  Also as noted above, it suffices to work with symmetric tensors $T\in S^m_1(\C^n)$, rather than Boson quantum states $\Psi\in\mathcal{B}^m(\C^n)$.

Let $T\in S^m_1(\C^n)$ be a symmetric tensor, with length $\langle T,T\rangle_2=1$.  Applying Proposition \ref{p.Schur}, taking inner products with $T$, we have
\begin{equation} \label{e.E.l.b.1} \int_{\mathbb{S}^1(\C^n)} \langle T,P_{\mx{v}}^{\tensor m}(T)\rangle_2\,d\mx{v} = \langle \mathcal{R}(T),T\rangle_2 = \textstyle{\binom{n+m-1}{m}}^{-1}\langle T,T\rangle_2 = \textstyle{\binom{n+m-1}{m}}^{-1}. \end{equation}
Now, let us compute the integrand.  By Proposition \ref{p.decomp.symmetric.rank1}, decompose $T = \sum_{j=1}^r \mx{w}_j^{\tensor m}$ for some vectors $\mx{w}_j\in\C^n$.  Then
\[ \langle T,P_{\mx{v}}^{\tensor m}(T)\rangle_2 = \sum_{j,k=1}^r \left\langle \mx{w}_j^{\tensor m}, \langle \mx{v},\mx{w}_k\rangle^m \mx{v}^{\tensor m}\right\rangle_2 = \sum_{j,k=1}^r \langle \mx{v}^{\tensor m},\mx{w}_k^{\tensor m}\rangle_2 \langle \mx{w}_j^{\tensor m}, \mx{v}^{\tensor m}\rangle_2.  \]
Distributing the sums inside the inner products, we therefore have
\[ \langle T,P_{\mx{v}}^{\tensor m}(T)\rangle_2 = \langle \mx{v}^{\tensor m},T\rangle_2\langle T,\mx{v}^{\tensor m}\rangle_2 = |\langle T,\mx{v}^{\tensor m}\rangle_2|^2. \]
Hence, \eqref{e.E.l.b.1} shows that the average value of the function $\mx{v}\mapsto |\langle T,\mx{v}^{\tensor m}\rangle_2|^2$ on the sphere is $\binom{n+m-1}{m}^{-1}$.  This is a continuous function, and hence by the mean value theorem for integrals, we conclude (from Definition \ref{e.spectral.norm}) that
\[ \|T\|_{\infty}^2 = \max_{\mx{v}\in\mathbb{S}^1(\C^n)} |\langle T,\mx{v}^{\tensor m}\rangle_2|^2 \ge \int_{\mathbb{S}^1(\C^n)} |\langle T,\mx{v}^{\tensor m}\rangle_2|^2\,d\mx{v} = \textstyle{\binom{n+m-1}{m}^{-1}}. \]
This holds true for every unit length symmetric tensor $T$, establishing the validity of \eqref{e.E.lower.bound}, and concluding the proof.  \end{proof}

\subsection{$\epsilon$-Nets on Boson States}

The proof of our main Theorem \ref{t.main} has two ingredients.  The first is a bound on the size of an $\e$-net for the Boson sphere $\mathcal{B}^1(\C^n)$.

\begin{definition} \label{d.e-net} Let $\mathcal{X}$ be a compact metric space, and let $\e>0$.  An {\bf $\e$-net} for $\mathcal{X}$ is a finite subset $\mathcal{N}\subseteq\mathcal{X}$ with the property that, for every $x\in\mathcal{X}$, there is a point $y\in\mathcal{N}$ with $d(x,y)<\e/2$; i.e.\ $\mathcal{X}$ is covered by $\e/2$-balls centered at points of $\mathcal{N}$. \end{definition}

\begin{lemma} \label{l.net} Let $0<\e<1$, and $n\in\N$.  The metric space $\mathcal{B}^1(\C^n)$ possesses an $\e$-net $\mathcal{N}(\e,n)$ with cardinality $\le K_n/\e^{2(n-1)}$, where $K_n\le 2^{n+1}n^n$.
\end{lemma}

\begin{proof} For any $\mx{v}=[v_1,\ldots,v_n]\in\C^n$, there is some $\zeta\in\C$ with $|\zeta|=1$ so that $\zeta v_1\in\R$.
Hence, if we produce an $\e$-net for the set $S'=\{\mx{v}\in\C^n\colon \|\mx{v}\|=1, v_1\in\R\}$ then the projection of this set
into $\mathcal{B}^1(\C^n)$ will still be an $\e$-net (since the projection is a contraction by \eqref{e.def.metric}).  The set $S'$ is the
unit sphere in $\R\times\C^{n-1}\cong\R^{2n-1}$, so we need to produce an $\e$-net for the sphere $\mathbb{S}^{2(n-1)}\subset \R^{2n-1}$.

To produce such a (crude) net, we circumscribe the sphere in a box in $\R^{2n-1}$.  We put grid points on the
$2(2n-1)$ faces, each of which is a unit box of dimension $2(n-1)$, and radially project them onto the sphere.  The radial projection inward is a contraction, so it suffices for the points on the surface of the box to form an $\e$-net.

Thus, we populate each of the $2(2n-1)$ faces with grid points, given grid spacing $\frac{1}{N}$ for $N$ to be chosen shortly.
The maximal distance between any two grid points is the box-diagonal $\sqrt{2(n-1)}/N$, and so the maximal distance
between any point in the face and its nearest grid point is $\le \sqrt{2(n-1)}/2N$.  So we must choose $N$ large enough
that $\sqrt{2(n-1)}/2N<\e/2$, i.e. $N> \sqrt{2(n-1)}/\e$.  The corresponding number of grid points is $N^{2(n-1)}$ per face, and with $2(2n-1)$ faces, this gives $2(2n-1)N^{(2(n-1)}$ grid points.

Since we only need to choose
$N$ any small amount larger than $\sqrt{2(n-1)}/\e$, we can construct an $\e$-net with any number of points larger than this
bound.  (This requires possibly choosing a non-integer $N$, but this can be done by having a grid with one row spaced closer
than all the others.)  Since $2(2n-1)(\sqrt{2(n-1)})^{2(n-1)} = 2^{n+1}(n-\frac12)(n-1)^{n-1} < 2^{n+1}n^n$, this completes the proof. \end{proof}

\begin{remark} \label{r.K2} The above is a blunt overestimate.  For example, an elementary argument using polar coordinates with $n=2$ shows that $K_2$ can be taken $\le \pi$, as opposed to $32$.  \end{remark}

\subsection{Concentration of Measure}

The second ingredient we need is a concentration of measure inequality which follows essentially unchanged from \cite{HP00,HR03}.

\begin{lemma} \label{l.Hiai.Petz} Let $d\in\N$, and fix a point $\mx{x}\in \mathbb{S}^{2d-1}\subset\C^d$.  Let $\mx{Z}$ be a Haar random variable on $\mathbb{S}^{2d-1}$.  Then for $0<\e<1$,
\[ \mathbb{P}(|\langle \mx{Z},\mx{x}\rangle|\ge\e) \le e^{-(2d-1)\e^2}. \]
\end{lemma}

We will also need the following basic norm inequality.

\begin{lemma} \label{l.tensor.deriv.ineq} For any $\mx{v},\mx{w}\in\C^n$,
\begin{equation} \label{e.tensor.deriv.ineq} \|\mx{v}^{\tensor m}-\mx{w}^{\tensor m}\|_2 \le m\cdot\max\{\|\mx{v}\|,\|\mx{w}\|\}^{m-1}\|\mx{v}-\mx{w}\|. \end{equation}
\end{lemma}

\begin{proof} This can be found as \cite[Theorem  3.9]{Fr82}.  \end{proof}

\begin{corollary} \label{c.1} Let $0<\e<1$, $n,m\in\N$, and let $\mathcal{N}(\e/m,n)$ be an $\e$-net for $\mathcal{B}^1(\C^n)$; cf.\ Lemma \ref{l.net}.  Let $\Psi\in\mathcal{B}^m(\C^n)$ be a Boson state with $\|\Psi\|_{\infty}\ge\e$.  Then there is some element $\mx{v}\in\mathcal{N}(\e/m,n)$ such that $|\langle \Psi, \mx{v}^{\tensor m}\rangle|\ge \e/2$. \end{corollary}

\begin{proof} We prove the contrapositive: suppose that $|\langle \Psi,\mx{v}^{\tensor m}\rangle_2|< \e/2$ for all $\mx{v}\in\mathcal{N}(\e/m,n)$.  By Banach's Theorem \ref{t.Banach}, there is some element $\mx{v}_0\in\mathcal{B}^1(\C^n)$ with $\|\Psi\|_{\infty} = |\langle \Psi,\mx{v}_0^{\tensor m}\rangle_2|$.  By definition, there is some element $\mx{v}\in\mathcal{N}(\e/m,n)$ with $\|\mx{v}_0-\mx{v}\|<\e/2m$.  Since $\mx{v}_0$ and $\mx{v}$ have length $1$, Lemma \ref{l.tensor.deriv.ineq} implies that $\|\mx{v}_0^{\tensor m}-\mx{v}^{\tensor m}\|_2<\e/2$.  Thus
\[ \|\Psi\|_{\infty} = |\langle\Psi,\mx{v}_0^{\tensor m}\rangle_2| \le |\langle\Psi,\mx{v}^{\tensor m}\rangle_2| + |\langle\Psi,\mx{v}_0^{\tensor m}-\mx{v}^{\tensor m}\rangle_2| < \frac{\e}{2} + \|\Psi\|_2\|\mx{v}_0^{\tensor m}-\mx{v}^{\tensor m}\|_2 < \e \]
where the penultimate inequality follows from the Cauchy--Schwarz inequality and the fact that $\Psi$ has length $1$.  \end{proof}

As an immediate consequence of Corollary \ref{c.1}, we see that, with respect to any probability measure on states $\Psi\in\mathcal{B}^m(\C^n)$,
\begin{equation*}  \mathbb{P}(\|\Psi\|_{\infty}\ge \e) \le \mathbb{P}\left(\max_{\mx{v}\in\mathcal{N}(\e/m,n)}|\langle\Psi,\mx{v}^{\tensor m}\rangle_2|\ge \e/2\right). \end{equation*}
From the union bound, it therefore follows that
\begin{equation} \label{e.P1}   \mathbb{P}(\|\Psi\|_{\infty}\ge \e) \le \#\mathcal{N}(\e/m,n)\cdot \max_{\mx{v}\in\mathcal{N}(\e/m,n)}\mathbb{P}(|\langle\Psi,\mx{v}^{\tensor m}\rangle_2|\ge \e/2). \end{equation}
Using Lemmas \ref{l.net} and \ref{l.Hiai.Petz}, we thus deduce the following.

\begin{proposition} \label{p.main} Let $0<\e<1$, and $n,m\in\N$.  Let $d_{n,m} = \binom{n+m-1}{m}$. With respect to the Haar measure on elements $\Psi\in\mathcal{B}^m(\C^n)$ (discussed in Section \ref{section Bosons}),
\[ \mathbb{P}(\|\Psi\|_{\infty}\ge \e) \le \frac{K_n m^{2(n-1)}}{\e^{2(n-1)}}e^{-(2d_{n,m}-1)\e^2/4}. \]
\end{proposition}

\begin{proof} The Haar measure on $\mathcal{B}^m(\C^n)$ is the push forward of the Haar measure on $\mathbb{S}^{2d_{n,m}-1}$.  Since the modulus of the inner product is invariant under a complex phase, we may apply Lemma \ref{l.Hiai.Petz} to conclude that, for any fixed $\Phi_0\in\mathcal{B}^m(\C^n)$,
\[ \mathbb{P}(|\langle\Psi,\Phi_0\rangle_2|\ge\e/2) \le e^{-(2d_{n,m}-1)\e^2/4}. \]
Applying this with $\Phi_0=\mx{v}^{\tensor m}$ for the maximizing $\mx{v}\in\mathcal{N}(\e/m,n)$ in \eqref{e.P1}, and applying Lemma \ref{l.net}, yields the result.  \end{proof}

Let us note that, in the qubit case $n=2$, by Remark \ref{r.K2}, we have $K_2\le\pi$, and in this case Proposition \ref{p.main} says
\begin{equation} \label{e.p.main.n=2} \mathbb{P}(\|\Psi\|_{\infty}\ge \e) \le \frac{\pi m^2}{\e^2}e^{-(2m+1)\e^2/4}. \end{equation}

\subsection{Proof of Theorem \ref{t.main}}

As above, let $d_{n,m}=\binom{n+m-1}{m}$.  Here we have $n$ fixed and $m$ potentially large.  Let
\[ \e^2 = 2n^3\frac{\log_2 d_{n,m}}{d_{n,m}}. \]
Since $m\mapsto d_{n,m}$ is an increasing function of $m$ for each fixed $n$, and since $x\mapsto \frac{\log_2 x}{x}$ tends to $0$ as $x\to\infty$, there is some $m_0(n)$ so that $2n^3\frac{\log_2d_{n,m}}{d_{n,m}}<1$ for $m\ge m_0(n)$. Applying Proposition \ref{p.main}, this yields
\begin{equation} \label{e.final.1} \mathbb{P}\left(\|\Psi\|_{\infty}^2 \ge 2n^3\frac{\log_2d_{n,m}}{d_{n,m}}\right) \le \frac{K_nm^{2(n-1)}}{(2n^3\frac{\log_2 d_{n,m}}{d_{n,m}})^{n-1}}e^{-(2d_{n,m}-1)\cdot \frac{n^3}{2}\frac{\log_2 d_{n,m}}{d_{n,m}}}. \end{equation}
Rearrange the upper bound as a product of three terms:
\begin{equation} \label{e.final.2} \frac{K_n}{(2n^3)^{n-1}}\cdot m^{2(n-1)}(d_{n,m})^{n-1} e^{-n^3\log_2 d_{n,m}}\cdot \frac{e^{\frac{n^3}{2}\frac{\log_2 d_{n,m}}{d_{n,m}}}}{(\log_2 d_{n,m})^{n-1}}. \end{equation}
For the first factor in \eqref{e.final.2}, Lemma \ref{l.net} gives
\begin{equation} \label{e.final.3} \frac{K_n}{(2n^3)^{n-1}} \le \frac{2^{n+1}n^n}{(2n^3)^{n-1}} = \frac{4}{n^{2n}}. \end{equation}
For the second factor in \eqref{e.final.2}, we begin by noting that
\[ d_{n,m} = \binom{m+n-1}{m} = \binom{m+n-1}{n-1} = \frac{(m+n-1)(m+n-2)\cdots(m+1)}{(n-1)!} \ge \frac{m^{n-1}}{(n-1)!}. \]
Thus $m^{2(n-1)} \le (n-1)!^2 (d_{n,m})^2$.  Stirling's approximation yields $(n-1)! \le e^{-(n-1)}n^{n-1/2}$, and so the second term in \eqref{e.final.2} is bounded above by
\[ (n-1)!^2 (d_{n,m})^{n+1} e^{-n^3\frac{\ln d_{n,m}}{\ln 2}} \le e^{-2(n-1)}n^{2n} (d_{n,m})^{n+1-\frac{1}{\ln 2}n^3} < e^{-2(n-1)}n^{2n}(d_{n,m})^{-n^3} \]
for $n\ge 2$.  Combining this with \eqref{e.final.3}, we see that the first two factors in \eqref{e.final.2} are bounded above by
\begin{equation} \label{e.final.4} 4e^{-2(n-1)} (d_{n,m})^{-n^3}. \end{equation}
For the third term in \eqref{e.final.2}, we've already chosen $m\ge m_0(n)$ so that $2n^3\frac{\log_2 d_{n,m}}{d_{n,m}} < 1$, and thus the exponential factor is $<\frac14$.  Hence, we have shown that, for $m\ge m_0(n)$,
\[ \mathbb{P}\left(\|\Psi\|_{\infty}^2 \ge 2n^3\frac{\log_2d_{n,m}}{d_{n,m}}\right) \le \frac{4e^{-2(n-1)+\frac14}}{(\log_2 d_{n,m})^{n-1}} (d_{n,m})^{-n^3}. \]
When $n\ge 2$, $4e^{-2(n-1)+\frac14} \le 4e^{-\frac74}<1$, and the denominator is $\ge 1$, so the upper bound is just $(d_{n,m})^{-n^3}$.  Taking $-\log_2$ of both sides of the inequality inside the $\mathbb{P}$ then verifies \eqref{e.main.1}.

If we simply evaluate \eqref{e.main.1} at $n=2$, we get the estimate
\[ \mathbb{P}\Big(E(\Psi) \ge \log_2(m+1)-\log_2\log_2(m+1)-4\Big) \ge 1-\frac{1}{(m+1)^{8}}, \qquad \text{for}\; m\ge 108. \]
(There is nothing sacrosanct about the exponent $n^3$; as the above analysis shows, at the expense of increasing the constant $m_0(n)$ and a larger additive constant inside $\mathbb{P}$, we can have any exponent we like, so the probability decays super-polynomially.)

To derive \eqref{e.main.2}, we make a more careful analysis using \eqref{e.p.main.n=2}.  Let $\alpha>0$, and set $\e^2 = \alpha\frac{\log_2(m+1)}{m+1}$.  We must choose $m$ large enough that this is $<1$.  Mimicking the preceding analysis, we find that
\begin{align*} \mathbb{P}\left(\|\Psi\|_{\infty}^2 \ge \alpha\frac{\log_2(m+1)}{m+1}\right) &\le \frac{\pi m^2(m+1)}{\alpha \log_2(m+1)} e^{-(2(m+1)-1)\frac{\alpha}{4}\frac{\log_2(m+1)}{m+1}} \\
&\le \frac{\pi}{\alpha}\cdot m^2(m+1)(m+1)^{-\frac{\alpha}{2\ln 2}}\cdot \frac{e^{\frac{\alpha}{4}\frac{\log_2(m+1)}{m+1}}}{\log_2(m+1)}.
\end{align*}
Due to the condition $\e^2<1$, the last term is $<e^{\frac14}$, and so we have the general estimate
\[ \mathbb{P}\left(\|\Psi\|_{\infty}^2 \ge \alpha\frac{\log_2(m+1)}{m+1}\right) \le \frac{\pi e^{\frac14}}{\alpha} m^{-\frac{\alpha}{2\ln 2}+3}, \qquad \text{provided}\quad \frac{\log_2(m+1)}{m+1} < \frac1\alpha. \]
Taking $\alpha=8$ which is larger than $2\pi e^{\frac14}$, we have $-\frac{8}{2\ln 2}+3<-\frac52$, and it is easy to verify that $\frac{\log_2(m+1)}{m+1}<\frac18$ for $m> 42$.  This justifies \eqref{e.main.2}, concluding the proof.

\subsection*{Acknowledgments} We would like to thank Bruce Driver and Brian Hall for useful conversations regarding the proof of Proposition \ref{p.Schur}.

\end{document}